\documentclass[journal]{IEEEtran}
\usepackage{amssymb, amsmath, amsthm}
\usepackage{graphicx, subfigure}
\usepackage{url, cite}
\usepackage{verbatim, booktabs}

\newtheorem{conj}{Conjecture}[section]

\newtheorem{theorem}{Theorem}

\newtheorem{subsec:coding}{subsec:coding}

\newtheorem{lemma}{Lemma}

\usepackage{color}

\newcommand{\ls}[1]  
   {\dimen0=\fontdimen6\the=#1\dimen0
    \advance\lineskip.5\fontdimen5\the\lineskip-\dimen0
    \lineskiplimit=.9\lineskip
    \baselineskip=\lineskip
    \advance\baselineskip\dimen0
    \normallineskip\lineskip
    \normallineskiplimit\lineskiplimit
    \normalbaselineskip\baselineskip
    \ignorespaces
   }

\begin{document}

\newcounter{MYtempeqncnt}
\title{Interference Alignment Improves the Capacity of OFDM Systems}


\author{Yi~Xu~\IEEEmembership{Student~Member,~IEEE},~Shiwen~Mao,~\IEEEmembership{Senior~Member,~IEEE},~and~Xin~Su~\IEEEmembership{Member,~IEEE}%
\thanks{
This work was presented in part at IEEE GLOBECOM 2012, Anaheim, CA, December 2012~\cite{Xu2012}.}
\thanks{Y. Xu and S. Mao are with the Department of Electrical and Computer Engineering, Auburn University, Auburn, AL 36849-5201, USA. X. Su is with the Research Institute of Information Technology, Tsinghua University, Beijing, P. R. China, 100084. Email: yzx0010@tigermail.auburn.edu, smao@ieee.org, suxin@tsinghua.edu.cn.
}
\thanks{{\bf Corresponding author}: Shiwen Mao. Email: smao@ieee.org, Tel: (334)844-1845, Fax: (334)844-1809. }
}

\maketitle


\begin{abstract}
Multi-user Orthogonal Frequency Division Multiplexing (OFDM) and Multiple Output Multiple Output (MIMO) have been widely adopted to enhance the system throughput and combat the detrimental effects of wireless channels. Recently, interference alignment was proposed to exploit interference to enable concurrent transmissions of multiple signals. In this paper, we investigate how to combine these techniques to further enhance the system throughput. We first reveal the unique characteristics and challenges brought about by using interference alignment in diagonal channels. We then derive a performance bound for the multi-user (MIMO) OFDM/interference alignment system under practical constraints, and show how to achieve this bound with a decomposition approach. The superior performance of the proposed scheme is validated with simulations. 
\end{abstract}


\pagestyle{headings}\thispagestyle{headings}

\maketitle

\section{Introduction} \label{sec:intro}


The past decade has witnessed drastic increase of wireless data traffic, largely due to the so-called ``smartphone revolution.''
As wireless data traffic is explosively increasing, the capacity of existing and future wireless networks will be greatly stressed. 
Many advanced wireless communication technologies, such as Orthogonal Frequency Division Multiplexing (OFDM) and Multiple Input Multiple Output (MIMO), are widely adopted to enhance the system capacity,  
while a huge amount of wireless access networks/base stations (BS) are deployed every year to accommodate the compelling need for larger capacity. Given the increasing wireless data volume and the more and more crowded BS deployment, interference is becoming the major factor that limits wireless network performance. 


Traditionally, interference is considered harmful and often treated as background noise. As the performance of point-to-point transmission techniques is approaching Shannon capacity, there is now considerable interest on exploiting interference for further capacity gains. It is shown that when interference is large, it can be decoded and canceled from the mixed signal (as in interference cancellation), while when interference is comparable, interference alignment can be adopted to enable concurrent transmissions.  Although interference is harmful in many cases, it could be beneficial for enhancing system throughput as long as the interference can be aligned.  We call this kind of interference beneficial interference. 

Interference alignment was first proposed in~\cite{Cadambe2008}, and the feasibility condition was investigated in~\cite{Yetis2010}.
Since in a large network, there are many users but limited dimensionality, the authors in~\cite{Shen2011} proposed the concept of ``best-effort'' interference alignment, and adopted an iterative algorithm to optimize it. However, how to use interference alignment to enhance the throughput in practical OFDM system was not fully considered. Shi et al. in~\cite{Shi2011} also considered the problem of interference alignment in multi-carrier interference networks. But it is not clear if the approach can be extended to the general case of a large number of subcarriers. 
In~\cite{Kafedziski13}, the authors proposed two schemes to adopt interference alignment in multi-cell MIMO OFDM systems. In the first scheme, interference alignment was used to remove the inter-cell interference, while zero-forcing precoding was used to suppress the intra-cell interference. In the second scheme, interference alignment was also used for inter-cell interference removal, while the OFDMA access scheme was applied for intra-cell interference cancellation. However, the fundamental performance bound of multi-user MIMO OFDM system with interference alignment has not been discussed. In~\cite{Zhangletter14}, the authors derived the necessary and sufficient conditions for the three-user OFDM system with interference alignment in the time domain. However, these conditions cannot be applied to system with more users or under other conditions. 
In \cite{Kerret} system with incomplete channel state information is considered. But in this paper we focus on the case where channel state information is complete and perfectly known at the transmitter.
Ayach et al. in~\cite{rhealth2010} investigated the feasibility problem MIMO-OFDM system with interference alignment over measured channels, while in this paper, we mainly concern about the theoretical bound when interference alignment is incorporated in the OFDM system. 

Interference alignment also finds many applications in practical wireless networks. 
In~\cite{debahtvt13}, a cognitive interference alignment scheme was presented to suppress both cross-tier and co-tier interferences in OFDM-based two-tier networks.
Interference alignment with limited feedback was discussed in \cite{Gao14} \cite{Kuchi}.
Multi-cell opportunistic interference alignment wass investigated in \cite{Leithon12}.
Authors in \cite{Castanheira} considered applying interference alignment to the HetNet (Heterogeneous Network) where both macrocell and small cell coexist.
In~\cite{yiiccn13} \cite{yitvt13}, the authors investigated the behaviors of primary users and secondary users under a Stackelberg game theory framework, where distributed interference alignment is adopted to enable spectrum leasing in the cognitive radio network. 
To achieve better error rate performance, a novel interference alignment based precoder design was presented in~\cite{xianggen12} for OFDM system.

There are also some existing studies that aim to adopt interference alignment in more advanced systems. 
In~\cite{Li2010}, the authors extended the traditional interference alignment scheme to a general algorithm for multi-hop mesh networks. The authors in~\cite{Gollakota2009} considered combining interference alignment and interference cancellation to further enhance the system throughput. In~\cite{Hadidy12}, the authors proposed to use multimode MIMO antennas instead of the typical omni-directional antennas to improve the performance of MIMO OFDM system with interference alignment, while in~\cite{Dimitrov2012}, the impact of antenna spatial correlation on the performance of interference alignment systems was investigated.

As claimed in~\cite{jafartut}, there are not many studies about interference alignment with structured channels. In~\cite{Xu2012}, the authors aimed to show how interference alignment works in OFDM system under practical constraints. To further address this problem, here in this paper, we consider the problem of incorporating interference alignment in multi-user (MIMO) OFDM systems. Specifically, we first examine the fundamental characteristics and practical constraints on adopting interference alignment in a multi-user OFDM system. We show that, for a $K$ user $N$ subcarrier OFDM system, $KN/2$ concurrent transmissions that is achievable for generic structureless channels~\cite{Cadambe2008}, cannot be achieved for a practical multi-user OFDM network with diagonal channels and a limited number of subcarriers. We then investigate effective schemes to exploit interference in multi-user OFDM systems. With an integer programming problem formulation, we derive the maximum efficiency of the Multi-user OFDM/interference alignment system. We also show how to achieve the maximum efficiency with a decomposition approach, and derive the closed-form precoding and decoding matrices. 
Finally, we extend the above analysis to the multiple antennas scenarios.
All the proposed schemes are evaluated with simulations and their superior performance is validated.  


\smallskip
\noindent
\underline{Notation}: in this paper, a capital bold symbol like $\textbf{H}$ denotes a matrix, a lower case symbol with an arrow on top like $\vec{v}$ 
denotes a vector, and a lower case letter like $v$ denotes a scalar. 
$[\cdot]^T$ means {\em transpose} and $[\cdot]^{-1}$ means {\em inversion}. 
$\textbf{H}_{ij}$ and $h_{ij}$ are the channel gain matrix and channel gain from the $i$-th transmitter to the $j$-th receiver, respectively. $\textbf{V}_i$ is the precoding matrix for transmitter $i$; 
$\vec{v}^j_i$ is the $j$-th column of $\textbf{V}_i$. $\textbf{U}_i$ denotes the interference cancellation matrix for the $i$-th receiver, while 
$\vec{u}^j_i$ is the $j$-th column of $\textbf{U}_i$. Let $h$, $v$, $u$ denote the entries of $\textbf{H}$, $\textbf{V}$, and $\textbf{U}$, respectively. 

Note that with these notations, the entries of $\textbf{H}_{ij}$ takes slightly different ordering from conventional ones. For instance, if transmitter $1$ and receiver $2$ are both equipped with $M$ antennas, the channel gain is:
\begin{equation} \label{eq:cgain}
\textbf{H}_{12}=
\begin{pmatrix}
h_{11} & h_{21} & \cdots & h_{M1} \\
h_{12} & h_{22} & \cdots & h_{M2} \\
\vdots & \vdots & \ddots & \vdots \\
h_{1M} & h_{2M} & \cdots & h_{MM}
\end{pmatrix}.
\end{equation}

The rest of this paper is organized as follows. Section~\ref{sec:bac} describes the background and preliminaries. 
Section~\ref{sec:muofdmia} investigates how to adopt interference alignment in multi-user OFDM system. 
Section~\ref{subsec:mumowia} extends the analysis to the multiple antennas scenario.
Simulation results are presented in Section~\ref{sec:simulation}. 
Section~\ref{sec:concl} concludes the paper. 

\section{Background and Preliminaries} \label{sec:bac}

\subsection{Orthogonal Frequency Division Multiplexing} \label{subsec:ofdm}
While higher data rates can be achieved
by reducing symbol duration, severe inter-symbol-interferences (ISI) will be caused over time dispersive channels. 
OFDM is an effective approach to allow transmissions at a high data rate  and combat the destructive effect of channel. By dividing the channel into narrow bands, in which the signal experiences flat fading, OFDM can effectively mitigate ISI and maintain high data rate transmissions. Interested reader are referred to~\cite{Hwang2009} and the references therein for details. 

\subsection{Multiple Input and Multiple Output} \label{subsec:mimo}
With the single antenna transmission technique being well developed, it is natural to extend to multiple antenna systems. 
The MIMO transmission techniques have been evolving rapidly since last decades. Generally speaking, multiple antennas or an antenna array can be used to attain the {\em diversity gain}, {\em multiplexing gain}, or {\em antenna gain}, and thereby reduce the system error rate, enhance the system throughput, or strengthen the signal to interference and noise ratio (SINR)~\cite{Mietzner2009}. 
Given $M_1$ transmitting antennas and $M_2$ receiving antennas, the maximum multiplexing gain is known to be $\min\{M_1,M_2\}$.
Throughout this paper, we assume that channel state information is perfectly known at each transmitter and receiver as in prior works~\cite{Cadambe2008}. 


\begin{figure} [!t] 
\center{\includegraphics[width=3.4in]{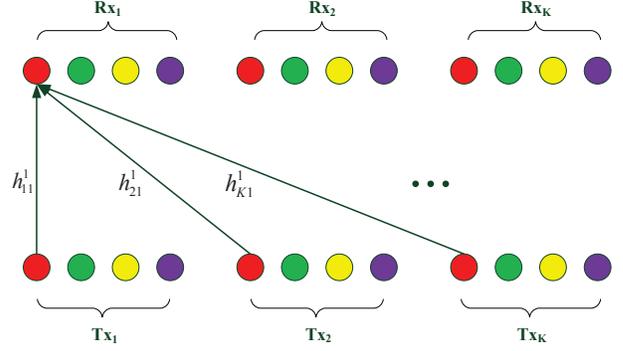}}
\caption{Multi-user OFDM using interference alignment.} 
\label{fig:mimoofdm1}
\end{figure}

\subsection{Interference Alignment} \label{subsec:ia}
It is shown in~\cite{Cadambe2008} that in a $K$ user wireless network, with $(n+1)^q+n^q$ symbol extensions, totally $K/2$ normalized {\em degrees of freedom} (DoF) can be achieved using interference alignment, where $q=(K-1)(K-2)-1$ and $n \in  \mathbb{N}$. In single antenna systems, the normalized DoF is 1. With interference alignment, the system throughput is enhanced by a factor of $K/2$ for $K \geq 2$. Note that there is no interference if there is only one user occupying the time or frequency resource. 

{\em Observation 1}:
The system throughput could be improved if alignable interference is introduced among users. 

This observation is useful for OFDM systems, where the channel gain matrix is diagonal. Since the gain of interference alignment is 
proportional to $K$, we should have more users transmit at the same time slot or frequency band if the transmitted vectors can be aligned. That is why we call this kind of interference {\em beneficial interference} in this paper.


\section{Multi-user OFDM with Interference Alignment} \label{sec:muofdmia}

In this section, we investigate the problem of interference alignment in multi-user OFDM systems. 
The system model is illustrated in Fig.~\ref{fig:mimoofdm1}.
We first examine fundamental characteristics and practical constraints, and then demonstrate how to exploit interference in multi-user OFDM systems. We derive the maximum throughput when interference alignment is adopted, as well as closed-form precoding and decoding matrices to achieve the maximum throughput. 

\subsection{Subcarriers versus Antennas} \label{subsec:rgdantur}

In traditional interference alignment, deploying multiple transmitting antennas allows us to precode data packets and align them at the receiver. Deploying multiple receiving antennas provides multidimensional signal space, so that interference can be aligned into a 
sub-signal space that is orthogonal to the desired signal. 
Therefore, deploying multiple antennas can provide the needed freedom in the signal space. 

In OFDM systems, we observe that subcarriers can function in similar ways as antennas in MIMO/interference alignment systems, since subcarriers could also provide multidimensional signal space.
To some extent, subcarriers can be regarded as a counterpart of antennas. So we could compress the interference at each receiver in no more than half of the subcarriers, and leave the other half subcarriers free from interference.

However, note that there is a distinguishing difference between the two systems: 
there is no cross-talk among different subcarriers in OFDM. 


\subsection{Precoding in OFDM} \label{subsec:precd}

The main idea of interference alignment is to compress the interference space to no more than half of the total received signal space at each receiver, leaving the remaining part of the space for desired signals~\cite{Cadambe2008}. 
This goal is achieved through precoding at every transmitter and zero forcing interference cancellation at every receiver.

In OFDM systems, data is transmitted over multiple carriers between transmitters and receivers, as shown in Fig.~\ref{fig:mimoofdm1}. Since in OFDM systems, subcarriers can also provide multidimensional signal space for the transmitter and receiver as multiple antennas, we could precode over multiple subcarriers to achieve interference alignment for OFDM system as multiple antennas for MIMO system. Suppose there are $N$ subcarriers. Ignoring noise, if there is no precoding, the received signal for each receiver $\vec{y}$ is an $N \times 1$ vector given by:
\begin{eqnarray}
\vec{y}=\textbf{H}\vec{x}, \label{eq:ofdm1}
\end{eqnarray}
where $\vec{x}$ is the desired signal in the form of an $N \times 1$ vector, and $\textbf{H}$ is the $N\times N$ channel gain matrix between the transmitter and receiver. 
Since different subcarriers have different frequencies, the channel gain matrix is {\em diagonal} if there is no severe frequency shift. It can be seen from later discussions that this property makes interference alignment in OFDM system quite different from the general channel case. 

Going one step further, we can precode the data before transmission. 
If $d$ packets are to be transmitted in an $N$ subcarrier OFDM system, an $N \times d$ {\em precoding matrix} $\textbf{V}$ could be used. The system equation is rewritten as follows.
\begin{eqnarray}
\vec{y}=\textbf{HV}\vec{x}. \label{eq:ofdm2}
\end{eqnarray}
If we let $d=N$ and $\textbf{V} = \textbf{I}_N$, where $\textbf{I}_N$ is an $N \times N$ identity matrix,~(\ref{eq:ofdm2}) is reduced to~(\ref{eq:ofdm1}). 

In general, we could control what to be transmitted on the subcarriers by adjusting the precoding matrix accordingly. For a single user single antenna OFDM system with $N$ subcarriers, the maximum number of packets can be transmitted is $N$. Note that, here $N$ is normalized by the QAM (Quadrature Amplitude Modulation) modulation level. However, inspired by the idea of interference alignment, we show that a throughput higher than $N$ can be achieved in the following subsections.

\subsection{Interference Alignment in a $K$-User OFDM System} \label{subsec:mmmimoofdm}

As discussed, we consider the problem of interference alignment in multiuser OFDM systems. Basically, we aim to 
answer the following questions.
\begin{enumerate}
	\item[(i)] What are the practical constraints for adopting interference alignment in such systems?
	\item[(ii)] What is the maximum throughput that can be achieved? 
	\item[(iii)] How to achieve the maximum throughput (i.e., deriving closed-form precoding and decoding matrices)?
\end{enumerate}

\smallskip
\subsubsection{Dependence of Precoding and Decoding Vectors in Diagonal Channels} \label{subsubsec:depofprede}

In this section, we show the difference on applying interference alignment between a diagonal channel and a general channel, as well as the challenges to adopt interference alignment in the former case. 

It was shown in~\cite{Yetis2010} that given $M_1$ transmitting antennas and $M_2$ receiving antennas in a $K$ user interference channel, the DoF for each user, denoted by $d$, must satisfy
\begin{eqnarray}
d \leq \frac{M_1+M_2}{K+1}. \label{eq:muofdmia97}
\end{eqnarray}
For example, given two transmitting and receiving antennas in a three-user interference channel,~(\ref{eq:muofdmia97}) indicates that each user could transmit one packet simultaneously. With a generic structureless channel, the throughput $Kd=3$ can be achieved 
as follows. 

At each receiver, we align the signals from the other two users. Recall the channel gain matrices as defined in~(\ref{eq:cgain}) and let the user $i$ signal be $\vec{v}_i$, $i=1, 2, 3$. It follows that
\begin{eqnarray}
\textbf{H}_{21}\vec{v}_2&=&\textbf{H}_{31}\vec{v}_3 \label{eq:muofdmia98} \\
\textbf{H}_{12}\vec{v}_1&=&\textbf{H}_{32}\vec{v}_3 \label{eq:muofdmia99} \\
\textbf{H}_{13}\vec{v}_1&=&\textbf{H}_{23}\vec{v}_2. \label{eq:muofdmia100}
\end{eqnarray}
Solving (\ref{eq:muofdmia98}), (\ref{eq:muofdmia99}) and (\ref{eq:muofdmia100}), we have
\begin{eqnarray}
\vec{v}_1&=&\mbox{eig}(\textbf{H}^{-1}_{12}\textbf{H}_{32}\textbf{H}^{-1}_{31}\textbf{H}_{21}\textbf{H}^{-1}_{23}\textbf{H}_{13}) \label{eq:muofdmia88}\\
\vec{v}_2&=&\textbf{H}^{-1}_{23}\textbf{H}_{13}\vec{v}_1 \label{eq:muofdmia87}\\
\vec{v}_3&=&\textbf{H}^{-1}_{32}\textbf{H}_{12}\vec{v}_1, \label{eq:muofdmia86}
\end{eqnarray}
where $\mbox{eig}(\textbf{A})$ stands for the eigenvector of matrix $\textbf{A}$. 

This scheme works well for generic structureless channels, 
but not for the case of diagonal channels. For instance, if 2 subcarriers (instead of two antennas) are used in OFDM, all the channel gain matrices in (\ref{eq:muofdmia88}), (\ref{eq:muofdmia87}) and (\ref{eq:muofdmia86})
are diagonal. Since the product of diagonal matrices is still diagonal, we have from (\ref{eq:muofdmia88}) that
$$\vec{v}_1=
\begin{pmatrix}
1 \\
0 
\end{pmatrix}
\; \mbox{or} \;
\begin{pmatrix}
0 \\
1 
\end{pmatrix}.
$$
If $\vec{v}_1=[1, 0]^T$, 
we derive $\vec{v}_2=[c_1, 0]^T$ from (\ref{eq:muofdmia87}) and $\vec{v}_3=[c_2, 0]^T$ from (\ref{eq:muofdmia86}), where $c_1$ and $c_2$ are scalars. To cancel the interference at receiver $1$, the {\em cancellation vector} $\vec{u}_1$ must be 
$\vec{u}_1=[0, c]^T$, where $c$ is also a scalar. 
However, the desired packet is also canceled since $\vec{u}_1$ is orthogonal to $\vec{v}_1$. 
Therefore, we cannot simultaneously transmit $3$ packets in this system.

The reason behind is that for a diagonal channel, its eigenvectors have only one nonzero entry. 
If we align interferences at receiver $r$ by letting $\textbf{H}_{jr}\vec{v}_j = \cdots =\textbf{H}_{ir}\vec{v}_i$, for $j \neq \cdots \neq i \neq r$, the precoding vectors are dependent to each other. Consequently, when interference is canceled at a receiver, the desired packet will also be canceled.

\smallskip
\subsubsection{Interference Alignment with Multi-user OFDM--Performance Bound} \label{subsubsec:uadiagnal}


It is shown in~\cite{Cadambe2008} that in a $K$ user system with $(n+1)^q+n^q$ symbol extensions, totally $K/2$ normalized 
DoF can be achieved using interference alignment, where $q=(K-1)(K-2)-1$ and $n \in \mathbb{N}$. In light of this result, one may think that $KN/2$ concurrent transmissions is achievable in a $K$-User, $N$ subcarrier OFDM system. However, we will show that this is unachievable for large $K$ in practical systems in the following.

It is worth noting that 
an assumption made in~\cite{Cadambe2008} is that the symbol extensions can be infinitely large. 
This assumption may not hold true in practical systems. Given a finite bandwidth, the number of subcarriers is the bandwidth divided by the subcarrier spacing. Typically, the value of subcarrier spacing is $10-20$ KHz. Then even for a 100 MHz bandwidth, we can have at most $10^4$ subcarriers. For instance, in 802.16m and LTE, the maximum number of IFFT is 2,048, and maximum number of effective subcarriers is 1,200. 

Therefore, the problem is to maximize system throughput given a finite number of subcarriers, 
denoted by $N_{max}$. It is shown in~\cite{Cadambe2008} that with $(n+1)^q+n^q$ symbol extensions, the total normalized DoF is
$[(n+1)^q+(K-1)n^q] / [(n+1)^q+n^q]$. So we aim to maximize $(n+1)^q+(K-1)n^q$ and have the following formulation.
\begin{eqnarray}
\max_{n,K} && (n+1)^q+(K-1)n^q \label{eq:mimoofdmmax} \\
      s.t. && q=(K-1)(K-2)-1  \label{eq:constraint1} \\
           && (n+1)^q+n^q \leq N_{max}, n \in \mathbb{N} \label{eq:constraint2} \\
           && K \geq 3, K \in \mathbb{N}. \label{eq:constraint3} 
\end{eqnarray}
The physical meaning of problem (\ref{eq:mimoofdmmax}) is that we try to maximize the unnormalized DoF given finite number of subcarriers. Note that all the variables are integers. 
Constraint~(\ref{eq:constraint2}) indicates that for practical OFDM systems, the number of subcarriers $N=(n+1)^q+n^q$ is upper bounded by $N_{max}$. Although this integer programming problem is NP-hard in general, by careful inspection, we can find the solution under practical constraints. 

In particular, we find the feasible region is very small for practical $N_{max}$ values. Also the objective value is monotone with respect to the two variables $n$ and $K$. 
In  problem (\ref{eq:mimoofdmmax}), assuming $K=5$, we have $q=11$ from (\ref{eq:constraint1}). For each value of $n$, we can derive the number of subcarriers needed, $N_{max}$, from (\ref{eq:constraint2}) for the problem to be feasible, as well as the throughput of the system (i.e., the objective value of (\ref{eq:mimoofdmmax})). The corresponding degree of freedom, $d$, is the ratio of the throughput and the number of subcarriers required. These numbers are presented in Table~\ref{tab:efficiencies}.

Table~\ref{tab:efficiencies} shows that 
if there are $K=5$ users, 
$2,049$ and $179,195$ subcarriers are needed when $n=1$ and $n=2$, respectively. As discussed, a practical system usually do not have more than $10^4$ subcarriers. So $n$ can only be $1$ in this case, with efficiency $d_{max} = 1.002$. Therefore, interference alignment is not useful in this case, since we can simply allow only one user to transmit over one time-slot or a particular frequency band to get $d = 1$ (i.e., single user OFDM).

If there are $K=6$ transmitters, we have $q=19$. 
Even if $n=1$, the number of subcarriers needed is $524,289$,
which is not feasible for practical systems. 
Since the number of subcarriers $(n+1)^{(K-1)(K-2)-1}+n^{(K-1)(K-2)-1}$ grows exponentially with $(K^2-3K+1)$, 
it can be readily concluded that $K$ cannot be more than 4 for interference alignment to be beneficial in multi-user OFDM systems. 

Since the 
objective value of (\ref{eq:mimoofdmmax}) 
is 
an monotone increasing function of $K$, the maximum feasible value $K=4$ is of particular interest. 
We have $q=5$ when $K=4$. 
Table~\ref{tab:efficiencies} also shows that under this condition, the maximum efficiency for practical system is 
$d_{max}=1.38$ for the practical case with at most $2,000$ subcarriers. 
When $K=3$, we have $q=1$. The objective function (\ref{eq:mimoofdmmax}) becomes $3n+1$, and the constraint (\ref{eq:constraint2}) becomes $2n+1 \leq N_{max}$. If the maximum number of subcarriers is $N_{max}=2,001$, the system achieves its maximum efficiency $d_{max}=1.4998$. 

The above analysis can be summarized as follows.

\begin{table}
\begin{center}
\caption{System Efficiency}
\label{tab:efficiencies}
\begin{tabular}{r|r|r|r}
\hline\hline
\multicolumn{4}{c}{When $K=5$ and $q=11$} \\ \hline
$n$ & \textit{No. of subcarriers} & \textit{No. of packets} & \textit{Normalized DoF $d$} \\
\hline
1 & 2,049 & 2,052 & 1.002\\
2 & 179,195 & 185,339 & 1.03\\
\hline\hline
\multicolumn{4}{c}{When $K=4$ and $q=5$} \\ \hline
$n$ & \textit{No. of subcarriers} & \textit{No. of packets} & \textit{Normalized DoF $d$} \\
\hline
1 & 33 & 35 & 1.06\\
2 & 275 & 339 & 1.23\\
3 & 1,267 & 1,753 & 1.38\\
4 & 4,149 & 6,197 & 1.49\\
\hline\hline
\multicolumn{4}{c}{When $K=3$ and $q=1$} \\ \hline
$n$ & \textit{No. of subcarriers} & \textit{No. of packets} & \textit{Normalized DoF $d$} \\
\hline
1 & 3 & 4 & 1.333\\
2 & 5 & 7 & 1.40\\
3 & 7 & 10 & 1.429\\
4 & 9 & 13 & 1.444\\
100 & 201 & 301& 1.498 \\
1000 & 2001 & 3001& 1.4998 \\
\hline\hline
\end{tabular}
\end{center}
\end{table}

\smallskip 
\begin{conj} \label{cj:1}
For a practical multi-user OFDM system with number of subcarriers less than $2,002$, the maximum efficiency is $d_{max}=1.4998$, which is achieved when there are $K=3$ users using $N=2,001$ subcarriers. 
\end{conj}

However, in the later discussions, we will show that this conjecture does not hold true.

\subsubsection{Interference Alignment with Multi-user OFDM--Realization} \label{subsubsec:wecannot}

It is shown in~\cite{Cadambe2008} how to design the precoding matrices to transmit $3n+1$ packets over $2n+1$ symbol extensions in a three-user interference channel (i.e., for a three-user system, we have $q=1$ and $N=(n+1)^q+n^q=2n+1$). We will derive the precoding/decoding procedure for interference alignment with multi-user OFDM and prove its efficacy in this section.  

The precoding matrices proposed in~\cite{Cadambe2008} for the case of three users are as follows.
\begin{eqnarray}
&& \textbf{V}_1=\textbf{A}  \label{eq:origin1} \\
&& \textbf{V}_2=\textbf{H}_{23}^{-1} \textbf{H}_{13} \textbf{C} \label{eq:origin2} \\ 
&& \textbf{V}_3=\textbf{H}_{32}^{-1} \textbf{H}_{12} \textbf{B}, \label{eq:origin3}
\end{eqnarray}
where
\begin{eqnarray}
&& \textbf{A}=[\vec{w}~\textbf{T}\vec{w}~\textbf{T}^2\vec{w}~\cdots~\textbf{T}^n\vec{w}] \label{eq:origin4} \\ 
&& \textbf{B}=[\textbf{T}\vec{w}~\textbf{T}^2\vec{w}~\cdots~\textbf{T}^n\vec{w}] \label{eq:origin5} \\ 
&& \textbf{C}=[\vec{w}~\textbf{T}\vec{w}~\textbf{T}^2\vec{w}~\cdots~\textbf{T}^{n-1}\vec{w}] \label{eq:origin6} \\ 
&& \textbf{T}=\textbf{H}_{21} \textbf{H}_{12}^{-1} \textbf{H}_{32} \textbf{H}_{23}^{-1} \textbf{H}_{13} \textbf{H}_{31}^{-1} \label{eq:origin7} \\ 
&& \vec{w}=[1~1~\cdots~1]^T. \label{eq:origin8}
\end{eqnarray}
Thus, the received signal at receiver $1$ is:
\begin{eqnarray} \label{eq:muofdmia68}
\vec{y}_1 =  \textbf{H}_{11}\textbf{V}_1\vec{x}_1+\textbf{H}_{21}\textbf{V}_2\vec{x}_2+\textbf{H}_{31}\textbf{V}_3\vec{x}_3. 
\end{eqnarray}

In the general case, since the data streams are independent of each other, the received mixed signal spans $3n+1$ dimensions of the space. 
In interference alignment with multi-user OFDM, the received signal spans only $2n+1$ dimensions of space. Solving these $2n+1$ equations will yield the desired packets. However, the challenge is, if $2n+1$ is too large, we may not be able to solve these equations efficiently (as can be seen from the later discussions). This problem can be addressed with a decomposition approach as given in the following theorem.

\smallskip 
\begin{theorem} \label{th:2}
For an $N$ subcarrier OFDM system, we can divide
the subcarriers into $\lfloor N/(2n+1) \rfloor$ groups, where $n \in \mathbb{N}$, and precode and decode the groups separately to achieve the interference alignment gain.
\end{theorem}
\smallskip 
\begin{proof}
Recall that the channel gain matrix in OFDM is diagonal. Generally, if every user tries to transmit $d$ packets over the $N$ subcarriers, we have
$$\textbf{HV}=
\begin{pmatrix}
h_1 & 0 & \cdots & 0 \\
0   & h_2& \cdots & 0 \\
\vdots & \vdots & \ddots & \vdots \\
0   & 0 & \cdots & h_N
\end{pmatrix}
\begin{pmatrix}
v_{11}  & \cdots & v_{1d} \\
v_{21}  & \cdots & v_{2d} \\
\vdots  & \ddots & \vdots \\
v_{N1}  & \cdots & v_{Nd}
\end{pmatrix}.
$$
The precoding vectors must satisfy the conditions given in (\ref{eq:origin1})-(\ref{eq:origin8}). 
Let the precoding matrix assume the following form.
\begin{equation}
\textbf{V}= 
\begin{pmatrix}
\widetilde{\textbf{V}}_1 & 0 & \cdots &0\\
0 & \widetilde{\textbf{V}}_2 & \cdots &0\\
\vdots &\vdots &\ddots &\vdots\\
0 & 0 & \cdots &\widetilde{\textbf{V}}_g\\
\end{pmatrix},
\label{eq:muofdmia85}
\end{equation}
where $g=N/(2n+1)$ is the number of groups and $\widetilde{\textbf{V}}_i$ is the precoding matrix for group $i$ 
with dimensions $(2n+1) \times (n+1)$ or $(2n+1) \times n$ (i.e., user 1 sends $(n+1)$ packets, 
and each of the other users sends $n$ packets over $(2n+1)$ subcarriers.)  
Without loss of generality, we assume $N$ is dividable by $2n+1$. 
Rewriting $\textbf{H}$ in the form of multiple diagonal sub-matrices with the same dimensions, we have
\begin{equation}
\textbf{HV}= 
\begin{pmatrix}
\widetilde{\textbf{H}}_1\widetilde{\textbf{V}}_1 & 0 & \cdots &0\\
0 & \widetilde{\textbf{H}}_2\widetilde{\textbf{V}}_2 & \cdots &0\\
\vdots &\vdots &\ddots &\vdots\\
0 & 0 & \cdots &\widetilde{\textbf{H}}_g\widetilde{\textbf{V}}_g\\
\end{pmatrix}.
\label{eq:muofdmia67}
\end{equation}

For instance, when $N=6$ and $n=1$, we have for transmitter $1$
\begin{equation}
\textbf{HV}= 
\begin{pmatrix}
h_1v_{11} & h_1v_{12} &0 &0\\
h_2v_{21} & h_2v_{22} &0 &0\\
h_3v_{31} & h_3v_{32} &0 &0\\
0 & 0&h_4v_{41} &h_4v_{42} \\
0 & 0&h_5v_{51} &h_5v_{52} \\
0 & 0&h_6v_{61} &h_6v_{62} \\
\end{pmatrix}.
\label{eq:muofdmia66}
\end{equation}
If there are 3 users, we can let $\textbf{H}_{21}\textbf{V}_2=\textbf{H}_{31}\textbf{V}_3$ at receiver $1$ to get
$$
\begin{pmatrix}
h^{(1)}_{21}v^{(1)}_2 & 0 & \cdots & 0 \\
h^{(2)}_{21}v^{(2)}_2 & 0 & \cdots & 0 \\
h^{(3)}_{21}v^{(3)}_2 & 0 & \cdots & 0 \\
0 & h^{(4)}_{21}v^{(4)}_2 & \cdots & 0 \\
0 & h^{(5)}_{21}v^{(5)}_2 & \cdots & 0 \\
0 & h^{(6)}_{21}v^{(6)}_2 & \cdots & 0 \\
\vdots & \vdots & \ddots & \vdots \\
0 & 0 & \cdots & h^{(N-2)}_{21}v^{(N-2)}_2\\
0 & 0 & \cdots & h^{(N-1)}_{21}v^{(N-1)}_2\\
0 & 0 & \cdots & h^{(N)}_{21}v^{(N)}_2
\end{pmatrix} \nonumber
$$
$$
= 
\begin{pmatrix}
h^{(1)}_{31}v^{(1)}_3 & 0 & \cdots & 0 \\
h^{(2)}_{31}v^{(2)}_3 & 0 & \cdots & 0 \\
h^{(3)}_{31}v^{(3)}_3 & 0 & \cdots & 0 \\
0 & h^{(4)}_{31}v^{(4)}_3 & \cdots & 0 \\
0 & h^{(5)}_{31}v^{(5)}_3 & \cdots & 0 \\
0 & h^{(6)}_{31}v^{(6)}_3 & \cdots & 0 \\
\vdots & \vdots & \ddots & \vdots \\
0 & 0 & \cdots & h^{(N-2)}_{31}v^{(N-2)}_3\\
0 & 0 & \cdots & h^{(N-1)}_{31}v^{(N-1)}_3\\
0 & 0 & \cdots & h^{(N)}_{31}v^{(N)}_3
\end{pmatrix}, \nonumber
$$
which indicates:
\begin{equation}
\begin{pmatrix}
h^{(i)}_{21}v^{(i)}_2 \\
h^{(i+1)}_{21}v^{(i+1)}_2 \\
h^{(i+2)}_{21}v^{(i+2)}_2 
\end{pmatrix}
=
\begin{pmatrix}
h^{(i)}_{31}v^{(i)}_3 \\
h^{(i+1)}_{31}v^{(i+1)}_3 \\
h^{(i+2)}_{31}v^{(i+2)}_3 
\end{pmatrix},
\; i=1, 4, \cdots, N-2. 
\end{equation}

Since the above conditions can also be obtained by separately encoding the $N/(2n+1)$ groups of subcarriers, we could decompose the problem into a number of subproblems, one for each group, and precode and decode the groups separately. 

It remains to show how to decode the packets for this scheme. Without loss of generality, we also assume $K=3$. If this scheme is adopted, 
each time we sequentially take out $2n+1$ subcarriers. The received signal at receiver $1$ is: 
\begin{eqnarray} \label{eq:muofdmia84}
\vec{y}_1 & = &\textbf{H}_{11}\textbf{V}_1\vec{x}_1+\textbf{H}_{21}\textbf{V}_2\vec{x}_2+\textbf{H}_{31}\textbf{V}_3\vec{x}_3 \nonumber \\
&=&\textbf{H}_{11}\textbf{V}_1\vec{x}_1+\textbf{H}_{21}\textbf{H}^{-1}_{23}\textbf{H}_{13}\textbf{C}\vec{x}_2+\textbf{H}_{31}\textbf{H}^{-1}_{32}\textbf{H}_{12}\textbf{B}\vec{x}_3\nonumber \\
&=&\textbf{H}_{11}\textbf{V}_1\vec{x}_1+\textbf{H}_{21}\textbf{H}^{-1}_{23}\textbf{H}_{13}\textbf{C}\vec{x}_2+\textbf{H}_{31}\textbf{H}^{-1}_{32}\textbf{H}_{12}\textbf{T}\textbf{C}\vec{x}_3\nonumber \\
&=&\textbf{H}_{11}\textbf{V}_1\vec{x}_1+\textbf{H}_{21}\textbf{H}^{-1}_{23}\textbf{H}_{13}\textbf{C}\vec{x}_2+\textbf{H}_{21}\textbf{H}^{-1}_{23}\textbf{H}_{13}\textbf{C}\vec{x}_3\nonumber \\
& = &
\textbf{H}_{11}\textbf{V}_1\vec{x}_1+\textbf{H}_{21}\textbf{H}^{-1}_{23}\textbf{H}_{13}\textbf{C}(\vec{x}_2+\vec{x}_3)
\nonumber \\
& = &
(\textbf{H}_{11}\textbf{V}_1~\textbf{H}_{21}\textbf{V}_2) \cdot 
\begin{pmatrix}
\vec{x}_1~\vec{x}_2+\vec{x}_3
\end{pmatrix}^T \hspace{-0.05in}. 
\end{eqnarray}
Taking the inverse of matrix $(\textbf{H}_{11}\textbf{V}_1~\textbf{H}_{21}\textbf{V}_2)$ and discard the packets from transmitters $2$ and $3$, we can recover the desired packets $\vec{x}_1$. Note that we exploit the {\em commutative} property of diagonal matrices in (\ref{eq:muofdmia84}).

At receiver $2$, the received signal is:
\begin{eqnarray} \label{eq:muofdmia65}
\vec{y}_2 & = &\textbf{H}_{12}\textbf{V}_1\vec{x}_1+\textbf{H}_{22}\textbf{V}_2\vec{x}_2+\textbf{H}_{32}\textbf{V}_3\vec{x}_3 \nonumber \\
& = & 
\textbf{H}_{12}(\vec{w}~\textbf{B})\vec{x}_1+\textbf{H}_{22}\textbf{V}_2\vec{x}_2+\textbf{H}_{12}\textbf{B}\vec{x}_3
\nonumber \\
& = &
\textbf{H}_{12}\vec{w}{x}^{(1)}_1+\textbf{H}_{22}\textbf{V}_2\vec{x}_2+\textbf{H}_{12}\textbf{B}
\begin{pmatrix}
x^{(2)}_1+x^{(1)}_3 \\
\vdots\\
x^{(n+1)}_1+x^{(n)}_3
\end{pmatrix}
\nonumber \\
& = &
(\textbf{H}_{22}\textbf{V}_2~\textbf{H}_{12}\vec{w}~\textbf{H}_{12}\textbf{B}) \cdot \nonumber \\
&& 
\begin{pmatrix}
\vec{x}_2, x^{(1)}_1, x^{(2)}_1+x^{(1)}_3, \cdots, x^{(n+1)}_1+x^{(n)}_3
\end{pmatrix}^T \hspace{-0.05in}.
\end{eqnarray}
Taking the inverse of matrix $(\textbf{H}_{22}\textbf{V}_2~\textbf{H}_{12}\vec{w}~\textbf{H}_{12}\textbf{B})$, we get $\vec{x}_2$.

At receiver $3$, the received signal is:
\begin{eqnarray} \label{eq:muofdmia64}
\vec{y}_3 & = &\textbf{H}_{13}\textbf{V}_1\vec{x}_1+\textbf{H}_{23}\textbf{V}_2\vec{x}_2+\textbf{H}_{33}\textbf{V}_3\vec{x}_3 \nonumber \\
& = & 
\textbf{H}_{13}(\textbf{C}~\textbf{T}^{n}\vec{w})\vec{x}_1+\textbf{H}_{13}\textbf{C}\vec{x}_2+\textbf{H}_{33}\textbf{V}_3\vec{x}_3
\nonumber \\
& = &
\textbf{H}_{13}\textbf{C}
\begin{pmatrix}
x^{(1)}_1+x^{(1)}_2 \\
\vdots\\
x^{(n)}_1+x^{(n)}_2
\end{pmatrix}
+\textbf{H}_{13}\textbf{T}^{n}\vec{w}x^{(n+1)}_1+\textbf{H}_{33}\textbf{V}_3\vec{x}_3
\nonumber \\
& = &
(\textbf{H}_{33}\textbf{V}_3~\textbf{H}_{13}\textbf{C}~\textbf{H}_{13}\textbf{T}^n\vec{w}) \cdot \nonumber \\
&&
\begin{pmatrix}
\vec{x}_3, x^{(1)}_1+x^{(1)}_2, \cdots, x^{(n)}_1+x^{(n)}_2, x^{(n+1)}_1
\end{pmatrix}^T \hspace{-0.05in}. 
\end{eqnarray}
Taking the inverse of matrix $(\textbf{H}_{33}\textbf{V}_3~\textbf{H}_{13}\textbf{C}~\textbf{H}_{13}\textbf{T}^n\vec{w})$, we can decode $\vec{x}_3$. After decoding each group separately, we then combine the decoded data. The theorem is thus proved.
\end{proof}

\smallskip
Note that the proof of Theorem~\ref{th:2} also leads to an algorithm to achieve interference alignment gains for any large $N \in \mathbb{N}$.

\smallskip
\subsubsection{Practical Issue of Large Channel Variance} \label{subsubsec:another}

Here we examine another practical problem of adopting interference alignment for multi-user OFDM.

A necessary condition to achieve interference alignment in OFDM is that the channel gain is drawn from a continuous distribution. 
As a result, if the variance of the channel is large, some of the channel gains can be very small in certain conditions, 
while some other channel gains can be very large. 
When precoding over all the subcarriers, 
after 
taking the inverse of the channel gain matrix, some entry of the precoding matrix could be $10^4$ times (or even more) larger than some other ones. The result is that the power of one subcarrier could be $10^8$ times (or even more) larger than that of another subcarrier. Given certain power constraints, the error performance of the system will suffer from great degradation, which makes interference alignment less useful.



In our proposed scheme, if the channel variance is large, there is also a certain chance that some entries of $\textbf{T}$ can be 
much larger than the others, since $\textbf{T}=\textbf{H}_{21} \textbf{H}_{12}^{-1} \textbf{H}_{32} \textbf{H}_{23}^{-1} \textbf{H}_{13} \textbf{H}_{31}^{-1}=\textbf{H}_{21} \textbf{H}_{32} \textbf{H}_{13} \textbf{H}_{12}^{-1}  \textbf{H}_{23}^{-1} \textbf{H}_{31}^{-1}$.
If we precode and decode over large $n$, since the last column of $\textbf{V}_1$, $\textbf{V}_2$ and $\textbf{V}_3$ are all obtained by multiplying $\textbf{T}^n$, the situation could be further exacerbated. 
The consequences are as follows. 
\begin{enumerate}
\item[(i)] Since some of the entries can be extremely small, the decoding matrices can be close to singular. Thus the desired signal cannot be decoded. 
\item[(ii)] Even if the decoding matrices is invertible, due to the transmitter power constraint, the system error performance could be rather poor. 
\end{enumerate}
In fact, even if $n=1$, there is still a chance that some matrices are not invertible. These 
are the
reasons why we cannot precode and decode for large $N$. This issue also demonstrate the importance of the proposed decomposition theorem (see Theorem~\ref{th:2}).

Take $\textbf{V}_1$ for instance. The constraint is the power on one subcarrier cannot be $10^a$ (e.g., $a=3$) times larger than the power on another subcarrier. If the constraint is violated, the system is considered to be in the outage state.
Let 
\begin{equation}
\textbf{T}= 
\begin{pmatrix}
t_1 & 0 & \cdots &0\\
0 & t_2 & \cdots &0\\
\vdots &\vdots &\ddots &\vdots\\
0 & 0 & \cdots &t_{2n+1}\\
\end{pmatrix},
\label{eq:muofdmia63}
\end{equation}
where $t_i=h^{(i)}_{21} h^{(i)}_{32} h^{(i)}_{13} /(h^{(i)}_{12} h^{(i)}_{23} h^{(i)}_{31}),~i=1,2,\ldots,(2n+1)$. $ t_{1},t_{2},\ldots,t_{2n+1} $ can be regarded as $i.i.d$ (independent identically distributed) random variables. Let $t$ denote the common distribution of $ t_{1},t_{2},\ldots,t_{2n+1} $. Define $ t_{(1)},t_{(2)},\ldots,t_{(2n+1)} $ be the order statistics of $ t_{1},t_{2},\ldots,t_{2n+1} $ with $t_{(1)}=\min_i t_i$, $t_{(2n+1)}=\max_i t_i$. 

Let $\gamma=t_{(2n+1)}/t_{(1)}$. From (\ref{eq:origin1})-(\ref{eq:origin8}), we have $\gamma^{2n} \leq 10^a$, thus
\begin{equation}
\gamma \leq 10^{a/(2n)},
\label{eq:muofdmia62}
\end{equation}
which means $t_{(2n+1)}$ cannot be $10^{a/(2n)}$ times larger than $t_{(1)}$. 

On the other hand, since $\gamma_{max}=10^{a/(2n)}$, we have
\begin{equation}
1-\left( \Pr\left\{t \geq \frac{t_{(2n+1)}}{10^ \frac{a}{2n} }\right\} \right)^{2n+1} \leq \Pr \left\{t_{(1)} \leq \frac{t_{(2n+1)}}{\gamma}\right\} \leq 1.
\label{eq:muofdmia61}
\end{equation}

It can be seen that $\Pr \left\{ t \geq \frac{t_{(2n+1)}}{10^ \frac{a}{2n} } \right\}$ is a decreasing function of $n$. With the power of $2n+1$, $\Pr \left\{ t_{(1)} \leq t_{(2n+1)}/\gamma \right\}$ will quickly converge to $1$. That means, with large $n$, $P(t_{(2n+1)} \geq \gamma t_{(1)})=1$. 
Therefore, with large $n$ the constraint~(\ref{eq:muofdmia62}) will not be satisfied.

Next, we show how large $n$ could be for given constraint~(\ref{eq:muofdmia62}).
%
The joint probability density function (PDF ) of $t_{(1)}$ and $t_{(2n+1)}$ is found as follows.
\begin{eqnarray} 
f_{t_{(1)}t_{(2n+1)}}(x,y) = \frac{\partial^2 F_{t_{(1)}t_{(2n+1)}}(x,y)}{\partial x \partial y}, \label{eq:prob99}
\end{eqnarray}
where $F_{t_{(1)}t_{(2n+1)}}(x,y)$ is the joint cumulative distribution function (CDF) of $t_{(1)}$ and $t_{(2n+1)}$. By the definition of partial derivative, we have:
%
\begin{align} \label{eq:prob95}
   &f_{t_{(1)}t_{(2n+1)}}(x,y) \nonumber \\ 
=& \frac{\partial }{\partial y} \{ \lim_{\Delta x \to 0} [F_{t_{(1)}t_{(2n+1)}}(x+\Delta x,y)- \nonumber \\ 
&F_{t_{(1)}t_{(2n+1)}}(x,y)]/\Delta x \} \nonumber \\
=& \lim_{\Delta x \to 0, \Delta y \to 0} [F_{t_{(1)}t_{(2n+1)}}(x+\Delta x,y+\Delta y)- \nonumber \\ 
&F_{t_{(1)}t_{(2n+1)}}(x,y+\Delta y)-F_{t_{(1)}t_{(2n+1)}}(x+\Delta x,y)+ \nonumber \\ 
&F_{t_{(1)}t_{(2n+1)}}(x,y)]/(\Delta x \Delta y) \nonumber \\
=& \lim_{\Delta x \to 0, \Delta y \to 0} [ \Pr\{x \leq t_{(1)} \leq x+ \Delta x, t_{(2n+1)} \leq y+\Delta y \} - \nonumber \\ 
& \Pr\{x \leq t_{(1)} \leq x+\Delta x, t_{(2n+1)} \leq y \}]/(\Delta x \Delta y) \nonumber \\
=& \lim_{\Delta x \to 0, \Delta y \to 0} \Pr \{x \leq t_{(1)} \leq x+\Delta x, \nonumber \\  
&y \leq t_{(2n+1)} \leq y+ \Delta y )/(\Delta x \Delta y \}.
\end{align}

To calculate the probability of the last equality, for any $x < y$, we can divide the $x$ axis into five disjoint intervals as: $I_1=(-\infty,x)$, $I_2=(x,x+\Delta x)$, $I_3=(x+\Delta x,y)$, $I_4=(y,y+\Delta y)$ and $I_5=(y+\Delta y,\infty)$. For each $t_i$, the probability it falls into each interval can be calculated as follows.
\begin{eqnarray}
&&p_1=\Pr\{t_i \in I_1\}=F_t(x) \label{eq:prob94} \\
&&p_2=\Pr\{t_i \in I_2\}=F_t(x+\Delta x)-F_t(x) \label{eq:prob93} \\
&&p_3=\Pr\{t_i \in I_3\}=F_t(y)-F_t(x+\Delta x) \label{eq:prob92} \\
&&p_4=\Pr\{t_i \in I_4\}=F_t(y+\Delta y)-F_t(y) \label{eq:prob91} \\
&&p_5=\Pr\{t_i \in I_5\}=1-F_t(y+\Delta y). \label{eq:prob90}
\end{eqnarray}

To make $(x \leq t_{(1)} \leq x+\Delta x, y \leq t_{(2n+1)} \leq y+\Delta y)$ happen, the statistics $\{ t_1, t_2, \ldots, t_{2n+1}\}$ must have exactly $1$ sample falling into interval $I_2$, $1$ falling into interval $I_4$, $(2n-1)$ falling into interval $I_3$, and $0$ elsewhere, which is a multinomial problem. So we have:
\begin{eqnarray} 
\Pr\{x \leq t_{(1)} \leq x+\Delta x, y \leq t_{(2n+1)} \leq y+\Delta y \}= \nonumber \\ 
\binom{2n+1}{0,1,(2n-1),1,0}p_1^0 p_2^1 p_3^{(2n-1)} p_4^1 p_5^0. \label{eq:prob89}
\end{eqnarray} 
It follows that
%
\begin{eqnarray}  \label{eq:prob87}
&& f_{t_{(1)}t_{(2n+1)}}(x,y) \nonumber \\
&=& \lim_{\Delta x \to 0, \Delta y \to 0} \left\{\frac{(2n+1)!}{(2n-1)!} \frac{F_t(x+\Delta x)-F_t(x)}{\Delta x} \times  \right. \nonumber \\
&& \left. \frac{F_t(y+\Delta y)-F_t(y)}{\Delta y} \times [F_t(y)-F_t(x+\Delta x)]^{2n-1} \right\} \nonumber \\
&=& (2n+1)(2n)f_t(x)f_t(y)[F_t(y)-F_t(x)]^{2n-1}. 
\end{eqnarray}

Since $t_i=h^{(i)}_{21} h^{(i)}_{32} h^{(i)}_{13} /(h^{(i)}_{12} h^{(i)}_{23} h^{(i)}_{31}),~i=1,2,\ldots,2n+1$, and each $h^{(i)}$ is a random variable, the distribution of $t_i$ is difficult to be explicitly found. Here we continue our analysis by approximating $t_i$ as a Uniform distributed or Rayleigh distributed random variable.

If $t_i$ is approximated as a Uniform distributed random variable and $t_i \in (0,1)$, we have:
\begin{eqnarray}
&&\Pr\left\{t_{(1)} \leq \frac{t_{(2n+1)}}{\gamma} \right\}\\ \label{eq:prob86}
&=&\int_0^1 \int_0^{\frac{y}{\gamma}} (2n+1)(2n)(y-x)^{2n-1} dxdy.\\ \label{eq:prob85}
&=&\left[1-\left(1-\frac{1}{\gamma}\right)^{2n}\right]\\ \label{eq:prob84}
& \geq & 1-(1-10^{-\frac{a}{2n}})^{2n}, \label{eq:prob83}
\end{eqnarray}
where the last inequality is a direct result of~(\ref{eq:muofdmia62}). 
Taking derivative of~(\ref{eq:prob83}), it can be found that $P_{outage}=\Pr\left\{t_{(1)} \leq \frac{t_{(2n+1)}}{\gamma}\right\}$ is an increasing function of $n$. For $a=3$, if $n=1$, $P_{outage}=0.0622$; if $n=2$, $P_{outage}=0.5431$; and if $n=3$, $P_{outage}=0.8978$. For a system with many subcarriers, it indicates that we can only precode over $n=1$.

If $t_i$ is approximated as a Rayleigh distributed random variable with PDF $f(x \mid\sigma)=\frac{x}{\sigma^2}\exp(-\frac{x^2}{2\sigma^2}),~x \geq 0$, then 
\begin{eqnarray}
f_{t_{(1)}t_{(2n+1)}}(x,y)=(2n+1)(2n)\frac{xy}{\sigma^4}\exp\left(-\frac{x^2+y^2}{2\sigma^2}\right) \cdot  \nonumber \\  
\left(\exp\left(-\frac{x^2}{2\sigma^2}\right)-\exp\left(-\frac{y^2}{2\sigma^2}\right)\right)^{2n-1}.
\end{eqnarray}\label{eq:prob82}
There's no closed-form solution of $\Pr\left\{t_{(1)} \leq \frac{t_{(2n+1)}}{\gamma}\right\}$ in this case. The numerical results are shown in Fig.~\ref{fig:prob}. It can be seen that the conclusion still holds, i.e., we can only precode over $n=1$. 

\begin{figure} [!t] 
\center{\includegraphics[width=3.4in]{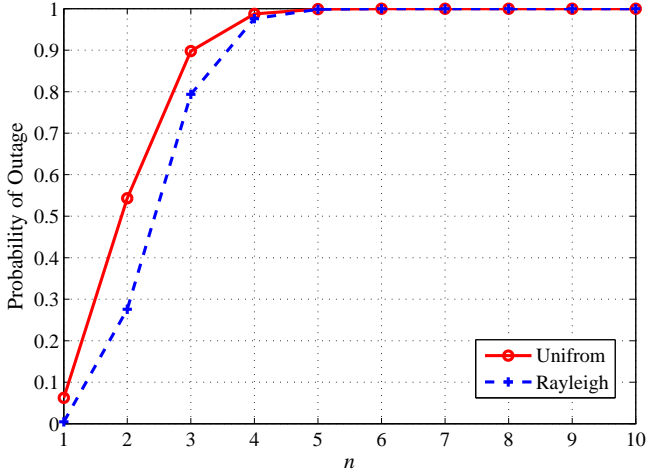}}
\caption{Probability of System Outage.} 
\label{fig:prob}
\end{figure}

Recall that Conjecture~\ref{cj:1} tells us $d_{max}=1.4998$ when $K=3$ and $n=1000$. Here we can see that this maximum DoF cannot be achieved under practical settings. So we have the following theorem.

\begin{theorem}\label{th:1s}
For a practical multi-user OFDM system with number of subcarriers less than $4149$, the maximum DoF is $d_{max}=1.33$, which is achieved when there are three transmitter/receiver pairs precoding over $3$ subcarriers each time.
\end{theorem}
\smallskip

\section{Multi-user MIMO OFDM with Interference Alignment} \label{subsec:mumowia}

In previous sections, we have considered applying interference alignment to OFDM systems. Since MIMO transmission technique can also be adopted to enhance the system throughput, we consider incorporating interference alignment to MIMO-OFDM systems in this section.

Suppose we have $M$ antennas at both the transmitter and receiver sides, and $N$ subcarriers in total. The signals received at the $i$-th receiver on subcarrier $n$ can be represented as: 
\begin{eqnarray} \label{eq:new96}
\vec{y}_i(n)  = \textbf{H}_{ii}(n)\textbf{V}_i(n)\vec{x}_i(n)+\sum_{j \neq i}\textbf{H}_{ji}(n)\textbf{V}_j(n)\vec{x}_j(n),
\end{eqnarray}
where $\textbf{H}_{ij}(n)$, $\textbf{V}_i(n)$, and $\vec{x}_i(n)$ are the channel matrix from transmitter $i$ to receiver $j$, precoding matrix at transmitter $i$, and data at transmitter $i$, respectively; and all of them are at subcarrier $n$. From~(\ref{eq:new96}), we can see that, the signals received can be represented as a matrix, with each column being the signals received from each subcarrier, i.e., $\textbf{Y}_i=\left[ \vec{y}_i(1)~\vec{y}_i(2)~\ldots~\vec{y}_i(n)\right]$. Or we could vectorize this matrix so that we get the following simpler form. 
\begin{eqnarray} \label{eq:new95}
\vec{y}_i  = \textbf{H}_{ii}\textbf{V}_i\vec{x}_i+\sum_{j \neq i}\textbf{H}_{ji}\textbf{V}_j\vec{x}_j.
\end{eqnarray}
Since each antenna pair could operate on any subcarrier and there is no crosstalk between subcarriers, the wireless channel $\textbf{H}_{ij}$ between transmitter $i$ and receiver $j$ is of the form as shown in~(\ref{eqn:leqn01}).

\begin{figure*}[!t]
\normalsize
\begin{equation}
\label{eqn:leqn01}
\textbf{H}_{ij}=
\begin{pmatrix}
h_{ij}^{1,1} & 0 & 0 & \cdots & h_{ij}^{N+1,1}& 0 & \cdots & h_{ij}^{(M-1)N+1,1} &0 & \cdots \\
0 &h_{ij}^{2,2}  & 0 & \cdots & 0 & h_{ij}^{N+2,2}& \cdots & 0& h_{ij}^{(M-1)N+2,2} & \cdots \\
\vdots & \vdots & \ddots & \vdots & \vdots & \ddots & \vdots & \vdots & \ddots & \vdots
\end{pmatrix}.
\end{equation}

\begin{equation}
\label{eqn:leqn02}
\textbf{H}_{ij}=
\begin{pmatrix}
h_{ij}^{1,1} & \cdots & h_{ij}^{M,1} & 0 &\cdots& 0 & \cdots & 0\\
\vdots & \ddots & \vdots  & 0 &\cdots& 0 & \cdots & 0\\
h_{ij}^{1,M} & \cdots &h_{ij}^{M,M}  & 0 &\cdots& 0 & \cdots & 0\\
\vdots & \vdots &\vdots  & \ddots &\cdots & \vdots & \vdots & \vdots\\
\cdots & \cdots &\cdots  & \cdots &\ddots & \cdots & \cdots & \cdots\\
0&\cdots &\cdots &0&\cdots&h_{ij}^{M(N-1)+1,M(N-1)+1}&\cdots&h_{ij}^{MN,M(N-1)+1}\\
\vdots& \vdots& \vdots& \vdots&\cdots &\vdots& \ddots& \vdots\\
0& \cdots& \cdots& 0&\cdots& h_{ij}^{M(N-1)+1,MN}& \cdots& h_{ij}^{MN,MN}
\end{pmatrix}.
\end{equation}

\hrulefill
\vspace*{4pt}
\end{figure*}


\smallskip 
\begin{theorem} \label{th:3}
For a MIMO-OFDM system with $N$ subcarriers and $M$ antennas at each transmitter and receiver side, we can divide
the subcarriers into $\lfloor N/(2n+1) \rfloor$ groups, where $n \in \mathbb{N}$, and precode and decode the groups separately to achieve the interference alignment gain.
\end{theorem}
\smallskip 
\begin{proof}
In Theorem~\ref{th:2}, we have actually established that for a system of diagonal channels, we could separately precode and decode each group of subcarriers. Now consider the case when all the devices are equipped with multiple antennas. We can still divide the subcarriers into different groups, then precode and decode them separately, since we are able to distinguish the signals from different antennas and different subcarriers. In other words, upon receiving a signal, the receiver has the knowledge of from which antenna and which subcarrier it gets the signal. So by properly adjusting the order of the data transmitted, the channel is essentially of the form in (\ref{eqn:leqn02}). We can readily identify that (\ref{eqn:leqn02}) is actually in the block diagonal form with the $i$-th block corresponding to the channels associated with the $i$-th subcarrier. Within each block, we have standard MIMO channels. Letting $V$, with dimension $MN \times d$, assume the form of~(\ref{eq:muofdmia85}), by similar arguments as in Theorem~\ref{th:2}, we could precode and decode the groups separately to achieve the interference alignment gain.
\end{proof}

\begin{lemma}\label{lm:1}
All the channel matrices and matrix $T$ are invertible.
\end{lemma}
\begin{proof}
As shown in (\ref{eqn:new97}), the inverse of a block matrix can be found by calculating the inverse of each block. Since for each block, we have a standard MIMO channel matrix and each of its entry is drawn from a continuous random distribution, each block is invertible with probability $1$. So each channel matrix is invertible. Since the product of invertible matrices is still invertible, according to (\ref{eq:origin7}), matrix $T$ is invertible.
\begin{equation}
\label{eqn:new97}
\begin{pmatrix}
B_{1}& 0 & 0 \\
0 & B_{2}&0 \\
0 & 0 & B_{3}
\end{pmatrix}^{-1}
=
\begin{pmatrix}
B_{1}^{-1}& 0 & 0 \\
0 & B_{2}^{-1}&0 \\
0 & 0 & B_{3}^{-1}
\end{pmatrix}.
\end{equation}
\end{proof}

\smallskip 
\begin{theorem} \label{th:4}
For a MIMO-OFDM system with $N$ subcarriers and $M$ antennas at each transmitter and receiver side, the maximum gain is $\frac{4}{3}M$.
\end{theorem}
\begin{proof}
According to Theorem~\ref{th:3}, we could precode and decode over groups of subcarrier. Also, according to our previous results, we can only precode and decode over $3$ subcarriers. So subcarrier-wise, the normalized DoF is $4/3$. 

We next show that $\frac{4}{3}M$ is the maximum achievable DoF. Firstly, we notice that by dividing the subcarriers into groups of $3$, taking $\textbf{H}_{11}$ for instance, it is transformed from~(\ref{eqn:new99}) to~(\ref{eqn:new98}). With the establishment of Lemma~\ref{lm:1}, following the proof of Theorem~\ref{th:2}, and replacing the scalars with blocks, we readily have the maximum gain of $\frac{4}{3}M$.
\begin{equation}
\label{eqn:new99}
\textbf{H}_{11}=
\begin{pmatrix}
h_{11}^{11} &0&0& h_{11}^{41} & 0 & 0 \\
0 & h_{11}^{22} & 0&0 & h_{11}^{52}&0 \\
0 & 0 & h_{11}^{33} &0 &0&h_{11}^{63}  \\
h_{11}^{14} & 0&0&h_{11}^{44} & 0 & 0 \\
0&h_{11}^{25}&0&0&h_{11}^{55}&0 \\
0&0&h_{11}^{36}&0&0&h_{11}^{66}
\end{pmatrix}.
\end{equation}

\begin{equation}
\label{eqn:new98}
\textbf{H}_{11}=
\begin{pmatrix}
h_{11}^{11} &h_{11}^{21}&0& 0 & 0 & 0 \\
h_{11}^{12} & h_{11}^{22} & 0&0 &0&0 \\
0 & 0 & h_{11}^{33} &h_{11}^{44} &0&0  \\
0 & 0&h_{11}^{34}&h_{11}^{44} & 0 & 0 \\
0&0&0&0&h_{11}^{55}&h_{11}^{65} \\
0&0&0&0&h_{11}^{56}&h_{11}^{66}
\end{pmatrix}.
\end{equation}

\end{proof}

We next show how to achieve this gain. 
We design $\textbf{V}_1$, $\textbf{V}_2$, and $\textbf{V}_3$ as follows.
\begin{eqnarray}
&& \textbf{V}_1=\textbf{A}  \label{eq:new89} \\
&& \textbf{V}_2=\textbf{H}_{23}^{-1} \textbf{H}_{13} \textbf{C} \label{eq:new88} \\ 
&& \textbf{V}_3=\textbf{H}_{32}^{-1} \textbf{H}_{12} \textbf{B}, \label{eq:new87}
\end{eqnarray}
where
\begin{eqnarray}
&& \textbf{A}=[\vec{w}~\textbf{T}\vec{w}~\textbf{T}^2\vec{w}~\cdots~\textbf{T}^{(n+1)M-1}\vec{w}] \\ \label{eq:new86}
&& \textbf{B}=[\textbf{T}^M\vec{w}~\textbf{T}^{M+1}\vec{w}~\cdots~\textbf{T}^{(n+1)M-1}\vec{w}] \\ \label{eq:new85}
&& \textbf{C}=[\textbf{T}^{M-1}\vec{w}~\textbf{T}^{M}\vec{w}~\cdots~\textbf{T}^{(n+1)M-2}\vec{w}] \\ \label{eq:new84}
&& \textbf{T}=\textbf{H}_{12}^{-1} \textbf{H}_{32} \textbf{H}_{31}^{-1} \textbf{H}_{21} \textbf{H}_{23}^{-1} \textbf{H}_{13}  \\ \label{eq:new83}
&& \vec{w}=[1~1~\cdots~1]^T. \label{eq:new82}
\end{eqnarray}

It can be observed that:
\begin{eqnarray}\label{eq:new81}
\textbf{A}&=&[\vec{w}~\textbf{T}\vec{w}~\cdots~\textbf{T}^{M-1}\vec{w}~\textbf{B}] \\ 
&=&[\vec{w}~\textbf{T}\vec{w}~\cdots~\textbf{T}^{M-2}\vec{w}~\textbf{C}~\textbf{T}^{(n+1)M-1}\vec{w}].
\end{eqnarray}

At receiver $1$, the received signals can be written as:
\begin{eqnarray} \label{eq:new92}
\vec{y}_1 & = &\textbf{H}_{11}\textbf{V}_1\vec{x}_1+\textbf{H}_{21}\textbf{V}_2\vec{x}_2+\textbf{H}_{31}\textbf{V}_3\vec{x}_3 \nonumber \\
& = & 
\textbf{H}_{11}\textbf{V}_1\vec{x}_1+\textbf{H}_{21}\textbf{H}^{-1}_{23}\textbf{H}_{13}\textbf{C}\vec{x}_2+\textbf{H}_{31}\textbf{H}^{-1}_{32}\textbf{H}_{12}\textbf{B}\vec{x}_3
\nonumber \\
& = &
\textbf{H}_{11}\textbf{V}_1\vec{x}_1+\textbf{H}_{21}\textbf{H}^{-1}_{23}\textbf{H}_{13}\textbf{C}\vec{x}_2+\textbf{H}_{31}\textbf{H}^{-1}_{32}\textbf{H}_{12}\textbf{T}\textbf{C}\vec{x}_3
\nonumber \\
& = &
\textbf{H}_{11}\textbf{V}_1\vec{x}_1+\textbf{H}_{21}\textbf{H}^{-1}_{23}\textbf{H}_{13}\textbf{C}\vec{x}_2+\textbf{H}_{21}\textbf{H}^{-1}_{23}\textbf{H}_{13}\textbf{C}\vec{x}_3
\nonumber \\
& = &
\textbf{H}_{11}\textbf{V}_1\vec{x}_1+\textbf{H}_{21}\textbf{H}^{-1}_{23}\textbf{H}_{13}\textbf{C}(\vec{x}_2+\vec{x}_3)
\nonumber \\
& = &
(\textbf{H}_{11}\textbf{V}_1~\textbf{H}_{21}\textbf{V}_2) \cdot 
\begin{pmatrix}
\vec{x}_1 \\
\vec{x}_2+\vec{x}_3
\end{pmatrix} \hspace{-0.05in}. 
\end{eqnarray}

For signals at receiver $2$, we have:
\begin{eqnarray} \label{eq:new91}
\vec{y}_2 & = &\textbf{H}_{12}\textbf{V}_1\vec{x}_1+\textbf{H}_{22}\textbf{V}_2\vec{x}_2+\textbf{H}_{32}\textbf{V}_3\vec{x}_3 \nonumber \\
& = & 
\textbf{H}_{12}(\vec{w}~\textbf{T}\vec{w}~\cdots~\textbf{T}^{M-1}\vec{w}~\textbf{B})\vec{x}_1+\textbf{H}_{22}\textbf{V}_2\vec{x}_2+\textbf{H}_{12}\textbf{B}\vec{x}_3
\nonumber \\
& = &
\textbf{H}_{12}(\vec{w}~\textbf{T}\vec{w}~\cdots~\textbf{T}^{M-1}\vec{w})
\begin{pmatrix}
x_1^{(1)} \\
\vdots \\
x_1^{(M)}
\end{pmatrix} \nonumber \\
&+&\textbf{H}_{12}\textbf{B}
\begin{pmatrix}
x_1^{(M+1)} \\
\vdots \\
x_1^{((n+1)M)}
\end{pmatrix}
+\textbf{H}_{22}\textbf{V}_2\vec{x}_2+\textbf{H}_{12}\textbf{B}
\vec{x}_3\nonumber \\
& = &
(\textbf{H}_{22}\textbf{V}_2~\textbf{H}_{12}(\vec{w}~\textbf{T}\vec{w}~\cdots~\textbf{T}^{M-1}\vec{w})~\textbf{H}_{12}\textbf{B}) \cdot \nonumber \\
&& 
\begin{pmatrix}
\vec{x}_2 \\ 
x^{(1)}_1\\ 
\vdots \\
x^{(M)}_1 \\ 
x^{(M+1)}_1+x^{(1)}_3\\ 
\vdots \\ 
 x^{((n+1)M)}_1+x^{(nM)}_3 
\end{pmatrix} \hspace{-0.05in}.
\end{eqnarray}

And similarly for signals at receiver $3$, we have:
\begin{eqnarray} \label{eq:new90}
\vec{y}_3 & = &\textbf{H}_{13}\textbf{V}_1\vec{x}_1+\textbf{H}_{23}\textbf{V}_2\vec{x}_2+\textbf{H}_{33}\textbf{V}_3\vec{x}_3 \nonumber \\
& = & 
\textbf{H}_{13}(\vec{w}~\textbf{T}\vec{w}~\cdots~\textbf{T}^{M-2}\vec{w}~\textbf{C}~\textbf{T}^{(n+1)M-1}\vec{w})\vec{x}_1 \nonumber \\
&+&\textbf{H}_{13}\textbf{C}\vec{x}_2+\textbf{H}_{33}\textbf{V}_3\vec{x}_3
\nonumber \\
& = &
\textbf{H}_{13} (\vec{w}~\textbf{T}\vec{w}~\cdots~\textbf{T}^{M-2}\vec{w})
\begin{pmatrix}
x^{(1)}_1 \\
\vdots\\
x^{(M-1)}_1
\end{pmatrix} +\textbf{H}_{33}\textbf{V}_3\vec{x}_3 \nonumber \\
&+&\textbf{H}_{13}\textbf{C}
\begin{pmatrix}
x_1^{(M)}\\
\vdots\\
x_1^{((n+1)M-1)}
\end{pmatrix}
+\textbf{H}_{13}\textbf{C}\vec{x}_2
\nonumber \\
&+&\textbf{H}_{13}\textbf{T}^{(n+1)M-1}\vec{w}x^{(n+1)M}_1 \nonumber \\
& = &
\begin{pmatrix}
\textbf{H}_{33}\textbf{V}_3 \\ 
\textbf{H}_{13}\textbf{C} \\
\textbf{H}_{13}(\vec{w}~\textbf{T}\vec{w}~\cdots~\textbf{T}^{M-2}\vec{w})\\
\textbf{H}_{13}\textbf{T}^{((n+1)M-1)}\vec{w}
\end{pmatrix}^T
\cdot \nonumber \\
&&
\begin{pmatrix}
\vec{x}_3\\
x_1^{(M)}+x_2^{(1)}\\
\vdots \\
x_1^{((n+1)M-1)}+x_2^{(nM)}\\
x_1^{(1)}\\
\vdots \\
x_1^{(M-1)}\\
x_1^{((n+1)M)}
\end{pmatrix} \hspace{-0.05in}. 
\end{eqnarray}

From~(\ref{eq:new92})--(\ref{eq:new90}), we can see that the desired signals are all free from interferences.


We can also calculate the probability of system outage when multiple antennas are deployed. 
So we need to find the probability of $\Pr\left\{ t_{(1)} \leq \frac{t_{((2n+1)M)}}{\gamma} \right\}$.
With similar arguments, the joint PDF of $t_{(1)}$ and $t_{((2n+1)M)}$ can be found as:
\begin{align} 
   & f_{t_{(1)}t_{((2n+1)M)}}(x,y) \nonumber \\
=& \lim_{\Delta x \to 0, \Delta y \to 0} P(x \leq t_{(1)} \leq x+\Delta x,  \nonumber \\ 
& y \leq t_{((2n+1)M)} \leq y+ \Delta y )/(\Delta x \Delta y) \nonumber \\
=& \binom{(2n+1)M}{0,1,(2n+1)M-2,1,0}p_1^0 p_2^1 p_3^{(2n+1)M-2} p_4^1 p_5^0. 
\end{align} 

If $t_i$ is approximated as a Uniform distributed variable in the range of $(0,1)$, the probability $\Pr\left\{ t_{(1)} \leq \frac{t_{((2n+1)M)}}{\gamma}\right\}$ can be found as follows.
\begin{align}
&\; \Pr\left\{ t_{(1)} \leq \frac{t_{((2n+1)M)}}{\gamma} \right\} =(2nM+M-1) \cdot\label{eq:new79} \\ 
&\; (2nM+M)\int_0^1 \int_0^{\frac{y}{\gamma}} (y-x)^{2nM+M-2} dxdy \label{eq:new78} \\ 
=&\; 1-\left(1-\frac{1}{\gamma}\right)^{(2n+1)M-1}\label{eq:new77} \\
\geq & \; 1-\left(1-10^{-\frac{a}{2(n+1)M-2}}\right)^{(2n+1)M-1}, \label{eq:new76}
\end{align}

If $t_i$ is approximated as a Rayleigh distributed variable, there is no closed-form solution for probability $\Pr \left\{ t_{(1)} \leq \frac{t_{((2n+1)M)}}{\gamma} \right\}$. The joint PDF of $t_{(1)}$ and $t_{((2n+1)M)}$ can be derived as:
\begin{align}
& f_{t_{(1)}t_{((2n+1)M)}}(x,y)=((2n+1)M)((2n+1)M-1)\frac{xy}{\sigma^4} \cdot \nonumber \\ 
& \exp\left(-\frac{x^2+y^2}{2\sigma^2}\right) \left[\exp\left(-\frac{x^2}{2\sigma^2}\right)-\exp\left(-\frac{y^2}{2\sigma^2}\right)\right]^{(2n+1)M-2} \label{eq:new75}
\end{align}

Figs.~\ref{fig:probmimouni} and~\ref{fig:probmimoray} illustrate the probabilities of system outage for Uniform and Rayleigh distributions, respectively. 
Here we also set $a=3$. So the power on one subcarrier cannot be $10^3$ times larger than the power on any other subcarrier.
We can see that when $n=1$ and $M=2$, the system outage probabilities are $0.8505$ and $0.2758$ for Uniform and Rayleigh distributions respectively. 
For $n=1$ and $M=3$, the outage probabilities are even higher: $0.9962$ and $0.7937$. 
For $n=2$ and $M=2$, the outage probabilities are unacceptably high as $0.9981$ and $0.9971$.
We can further see from simulations in the next section that, the higher outage probability is undesirable.

\begin{figure} [!t] 
\center{\includegraphics[width=3.4in]{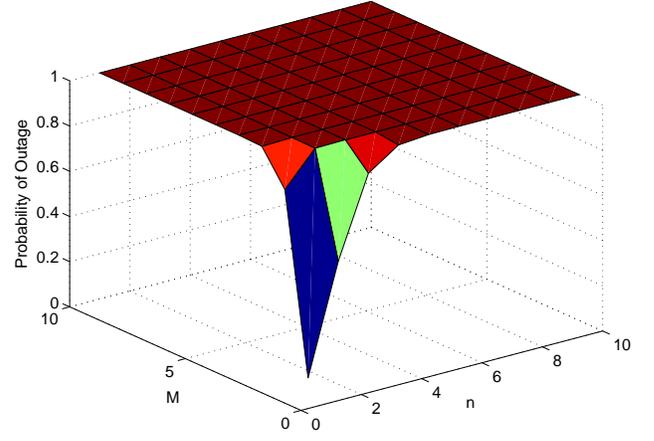}}
\caption{Probability of System Outage with multiple antennas for Uniform distribution.} 
\label{fig:probmimouni}
\end{figure}

\begin{figure} [!t] 
\center{\includegraphics[width=3.4in]{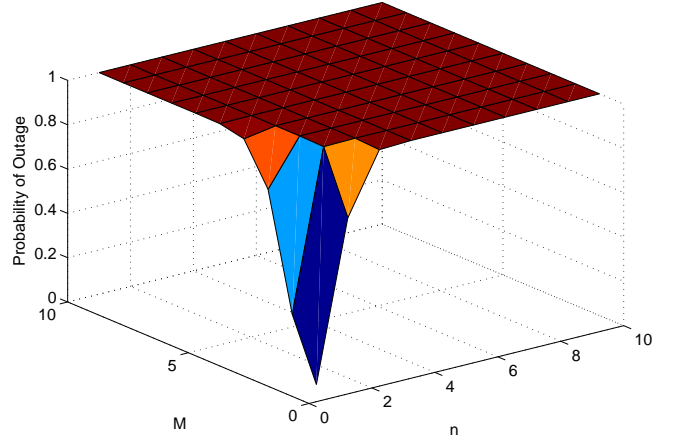}}
\caption{Probability of System Outage with multiple antennas for Rayleigh distribution.} 
\label{fig:probmimoray}
\end{figure}





\section{Simulation Study} \label{sec:simulation}




\begin{figure} [!t] 
\center{\includegraphics[width=3.4in]{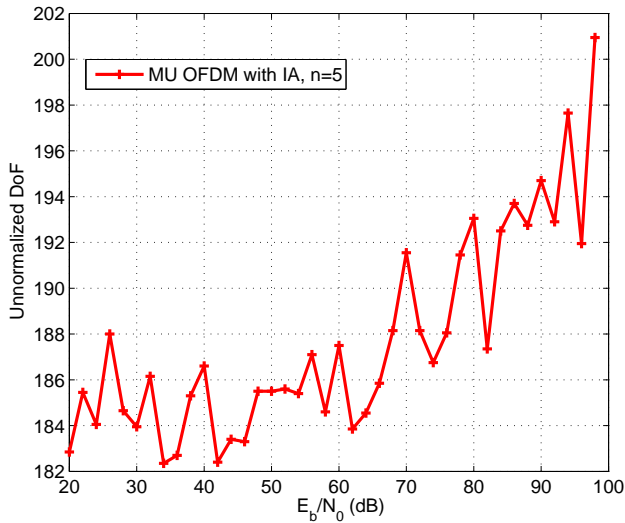}}
\caption{System throughput when $n=5$.} 
\label{fig:res00}
\end{figure}

\begin{figure} [!t] 
\center{\includegraphics[width=3.4in]{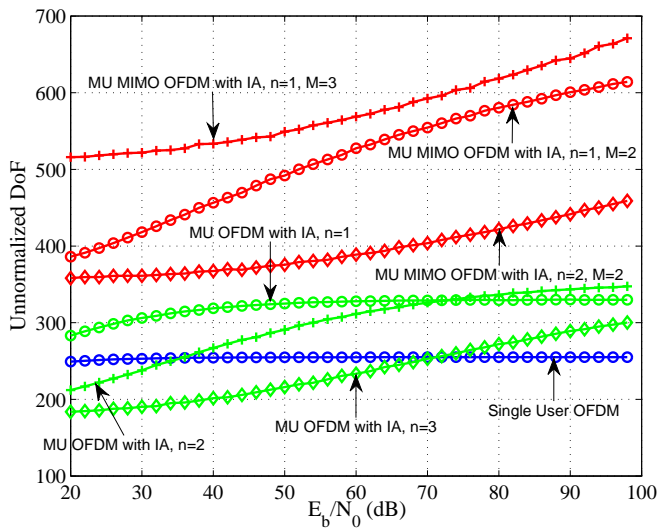}}
\caption{System throughput comparison when the channel variance is large.} 
\label{fig:res1n}
\end{figure}

\begin{figure} [!t] 
\center{\includegraphics[width=3.4in]{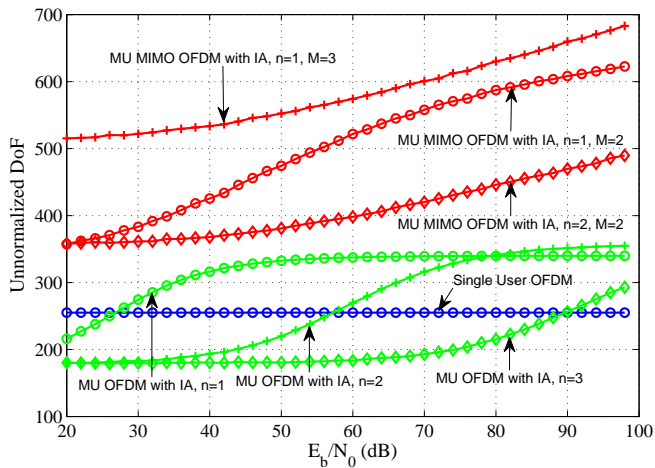}}
\caption{System throughput comparison when the channel variance is small.} 
\label{fig:res4n}
\end{figure}

Simulations are conducted to evaluate the performance of the proposed schemes and verify the benefits brought about by incorporating interference alignment in multi-user OFDM systems. 
We consider the case of $3$ users. The number of subcarriers is $255$. Each transmitter precodes over $(2n+1)M$ subcarriers. Block fading channels are used in the simulations, where channel gains are piece-wise constants for the duration of each time slot drawn from a certain distribution. BPSK is used as the modulation scheme. So we transmit $1$ bit on each subcarrier and we measure how many bits are successfully decoded at the receivers. In this way, we are essentially calculating the number of interference-free channels in the system (we call it {\em unnormalized DoF} hereafter).

Fig.~\ref{fig:res00} depicts the system throughput when $n=5$. We can see that the system performance is unstable. Comparing Fig.~\ref{fig:res00} with Fig.~\ref{fig:res1n} and Fig.~\ref{fig:res4n}, we could tell that since $n$ is too large, the system performance is degraded. This confirms our result that we could not precode over a large amount of subcarriers.

Fig.~\ref{fig:res1n} and Fig.~\ref{fig:res4n} illustrate the performances of different schemes when the channel is drawn from an uniform distribution on $\left[0,1\right]$ and $\left[0.9, 1\right]$, respectively. 
Comparing these two figures, we can see that when the channel variance is small, higher system throughput can be achieved. This conforms to our discussions about the precoding matrix in Section~\ref{subsubsec:another}. It can also be observed that the trends and comparative relationships are similar in Fig. 5 and Fig. 6.

We can see from Fig.~\ref{fig:res4n} that when $n=1$, multiuser OFDM with interference alignment can achieve an unnormalized DoF of  $339.98$. Compared to the highest throughput of single user OFDM of $255$, the DoF has been improved by a factor of approximately $1.33$ by incorporating interference alignment.
When $n=2$, we can see from both figures that the throughput of multiuser OFDM with interference alignment has degraded when the SNR is in the range $\left[0, 78\right]$ dB. That verifies our theorem that under certain power constraint, we can only precode over $3$ subcarriers. Same conclusions also hold for $n=3$ of multiuser OFDM with interference alignment, which exhibits poorer performance in the SNR range of $\left[20,100\right]$ dB.

For the case of multiuser MIMO OFDM with interference alignment, when $n=1$ with small channel variance, the highest unnormalized DoF is $622.7$, which is $2.44$ times of the unnormalized DoF of the single user OFDM system. The reason why it is slightly less than $2.66$ is also due to the big differences among the elements of the precoding matrices. For $n=2$ and $M=2$, we can see that the performance is worse than that of $n=1$ and $M=2$. When the devices are equipped with $3$ antennas, we let $n=1$ and precode over $3$ subcarriers. The highest unnormalized DoFs are $671.2$ and $683.208$ for large and small channel variance cases, respectively, which are $2.63$ and $2.68$ times of that of the single user OFDM system.
However,  the maximum gain is suppose to be 4 times the single user OFDM system. The performance degradation is also due to big difference among the elements of the precoding matrices.

\section{Conclusions} \label{sec:concl}

In this paper, we investigated the problem of how to exploiting interference in OFDM systems. 
We provided an analysis and developed effective schemes on incorporating interference alignment with multi-user (MIMO) OFDM to enhance system throughput. 
With an integer programming formulation, we derived the maximum efficiency for multi-user (MIMO) OFDM/interference alignment systems, and showed how to achieve the maximum efficiency under practical constraints. The performance of the proposed schemes were validated with simulations. The proposed decomposition algorithm and results of this paper may serve as guidance for practical OFDM system design.

%


\begin{thebibliography}{20}


\bibitem{Xu2012}
Y. Xu and S. Mao, ``On interference alignment in multi-user OFDM systems,'' in {\em Proc. IEEE GLOBECOM 2012}, Anaheim, CA, Dec. 2012, pp.5339--5344.

\bibitem{Cadambe2008}
V. R. Cadambe and S. A. Jafar, ``Interference alignment and degrees of freedom of the K-user interference channel,'' \textit{IEEE Transactions on Information Theory}, vol. 54, no. 8, pp. 3425--3441, 2008.

\bibitem{Yetis2010}
C. M. Yetis, T. Gou, S. A. Jafar, and A. H. Kayran,``On feasibility of interference alignment in MIMO interference networks,'' \textit{IEEE Transactions on Signal Processing}, vol. 58, no. 9, pp. 4771--4782, 2010.

\bibitem{Shen2011}
M. Shen, C. Zhao, X. Liang and Z. Ding, ``Best-effort interference alignment in OFDM systems with finite SNR,'' in \textit{Proc. IEEE ICC}, Kyoto, Japan, June 2011, pp. 1--6.

\bibitem{Shi2011}
C. Shi, R. A. Berry, and M. L. Honig, ``Interference alignment in multi-carrier interference networks,'' in \textit{Proc. IEEE ISIT}, Saint Petersburg, Russia, July 2011, pp. 1--5.

\bibitem{Kafedziski13}
V. Kafedziski and T. Javornik, ``Frequency-space interference alignment in multi-cell MIMO OFDM downlink systems,'' \textit{Proc. IEEE VTC Spring}, Dresden, Germany, June 2013, pp. 1--5.


\bibitem{Zhangletter14}
Q. Zhang, Q. Yong, J. Qin, A. Nallanathan, ``On the Feasibility of the CJ Three-User Interference Alignment Scheme for SISO OFDM Systems,'' \textit{IEEE Communications Letters}, to appear.


\bibitem{Kerret}
P. De Kerret, and D. Gesbert, ``Interference Alignment with Incomplete CSIT Sharing,'' \textit{IEEE Transaction on Wireless Communications}, vol. 13, iss. 5, pp. 2563--2573, Mar. 2014.

\bibitem{rhealth2010}
O. E. Ayach, S. W. Peters and R. W. Heath, Jr., ``The feasibility of interference alignment over measured MIMO-OFDM channels,'' \textit{IEEE Transactions on Vehicular Technology}, vol. 59, no. 9,  pp. 4309--4321, Nov. 2010.



\bibitem{debahtvt13}
M. Maso, M. Debbah, L. Vangelista, ``A Distributed Approach to Interference Alignment in OFDM-Based Two-Tiered Networks,'' \textit{IEEE Transactions on Vehicular Technology}, vol.62, iss.5, pp. 1935--1949, Jun. 2013.




\bibitem{Gao14}
H. Gao, T. Lv, S. Yang, and C. Yuen, ``Limited Feedback-Based Interference Alignment for Interfering Multi-Access Channels,'' \textit{IEEE Communications Letters}, vol. 18, iss. 4, pp. 540--543, Feb. 2014.

\bibitem{Kuchi}
K. Kuchi, ``Exploiting Spatial Interference Alignment and Opportunistic Scheduling in the Downlink of Interference-Limited Systems,'' vol. 63, iss. 6, pp. 2673--2686, Nov. 2013.

\bibitem{Leithon12}
H. Gao, J. Leithon, C. Yuen, and H. Suraweera, ``New Uplink Opportunistic Interference Alignment: An Active Alignment Approach,'' in \textit{Proc. IEEE WCNC 2013}, Shanghai, China, Apr. 2013, pp. 1--6.

\bibitem{Castanheira}
D. Castanheira, A. Silva, and A. Gameiro, ``Set Optimization for Efficient Interference Alignment in Heterogeneous Networks,'' \textit{IEEE Transactions on Wireless Communications}, vol. 13, iss. 10, pp. 5648--5660, May 2014.

\bibitem{yiiccn13}
Y. Xu and S. Mao, ``Distributed interference alignment in cognitive radio networks,'' in \textit{Proc. IEEE ICCCN}, Nassau, Bahamas, July/August 2013, pp. 1--7.

\bibitem{yitvt13}
Y. Xu and S. Mao, ``Stackelberg game for cognitive radio networks with MIMO and distributed interference alignment,'' \textit{IEEE Transactions on Vehicular Technology}, vol. 63, no. 2, pp. 879--892, Feb. 2014.

\bibitem{xianggen12}
Y. Jin, X.-G. Xia, ``An interference alignment based precoder design using channel statistics for OFDM systems with insufficient cyclic prefix,'' in \textit{Proc. IEEE GLOBECOM}, Anaheim, CA, Dec. 2012, pp. 3778--3782.

\bibitem{Li2010}
L. E. Li, R. Alimi, D. Shen, H. Viswanathan, and Y. R. Yang, ``A General Algorithm for Interference Alignment and Cancellation in Wireless Networks,'' in \textit{Proc. IEEE INFOCOM}, San Diego, CA, Mar. 2010, pp. 1--9.

\bibitem{Gollakota2009}
S. Gollakota, S. D. Perli and D. Katabi, ``Interference alignment and cancellation,'' in \textit{Proc. ACM SIGCOMM}, Barcelona, Spain, Aug. 2009, pp. 159--170.

\bibitem{Hadidy12}
M. El-Hadidy, M. El-Absi, L. Sit, ``Improved Interference Alignment Performance for MIMO OFDM Systems by Multimode MIMO Antennas,'' in \textit{Proc. IEEE InOWo}, Essen, Germany, Aug. 2012, pp. 1--5.

\bibitem{Dimitrov2012}
C. Le, E. Dimitrov, A. Anggraini, J. Peissig, and H.-P. Kuchenbecker, ``Effect of spatial correlation on MMSE-based interference alignment in a multiuser MIMO MB-OFDM system,'' in \textit{Proc. IEEE WiMob}, Barcelona, Spain, Oct. 2012, pp. 739--744.

\bibitem{jafartut}
S. Jafar, ``Interference Alignment: A New Look at Signal Dimensions in a Communication Network,'' \textit{Foundations and Trends in Communications and Information Theory}, vol.7, no.1, pp. 1--136, 2010.


\bibitem{Hwang2009}
T. Hwang, C. Yang, G. Wu, S. Li and Y. G. Li,``OFDM and its wireless applications: a survey,'' \textit{IEEE Transactions on Vehicular Technology}, vol. 58, no. 4, pp. 1673--1694, 2009.





\bibitem{Mietzner2009}
J. Mietzner, R. Schober, L. Lampe, W. H. Gerstacker and P. A. Hoeher,``Multiple-antenna techniques for wireless communications - a comprehensive literature survey,'' \textit{IEEE Communications Surveys \& Tutorials}, vol. 11, no. 2, pp. 87--105, 2009.






\end{thebibliography}

\end{document}